%% file: dcliques.tex
\documentclass[twocolumn,final]{svjour3}

\usepackage{graphicx}
\usepackage{balance}  
\usepackage{times}
\usepackage{amsmath,epsfig}
\usepackage{cite}

\usepackage[ruled,vlined,linesnumbered,commentsnumbered]{algorithm2e}
\usepackage{color}
\usepackage{booktabs}
\usepackage{multicol}
\usepackage{epsfig}
\usepackage{epstopdf}
\usepackage{url}
\usepackage{float}
\usepackage{subcaption}
\usepackage[labelfont=bf,skip=7pt]{caption}
\usepackage{changepage}
\usepackage{soul}
\usepackage{wrapfig}
\usepackage{ragged2e}

\newcommand{\remove}[1]{}

\newcommand{\cliques}{{\cal C}}

\newcommand{\sdiff}{{\tt IMCE}}

\newcommand{\ov}{{\tt OV}}
\newcommand{\stix}{{\tt STIX}}

\newcommand{\naive}{{\tt Naive}}
\newcommand{\mcmei}{{\tt MCMEI}}

\newcommand{\cand}{{\tt cand}}
\newcommand{\fini}{{\tt fini}}
\newcommand{\tomita}{{\tt TTT}\xspace}
\newcommand{\tomitaE}{{\tt TTTExcludeEdges}\xspace}
\newcommand{\pivot}{{\tt pivot}}
\newcommand{\ext}{{\tt ext}}
\newcommand{\csnew}{{\tt IMCENewClq}\xspace}
\newcommand{\csnewttt}{{\tt FastIMCENewClq}\xspace}
\newcommand{\cssub}{{\tt IMCESubClq}\xspace}
\newcommand{\dec}{{\tt Decremental}}

\newtheorem{observation}{Observation}

\DontPrintSemicolon
\SetAlgoLined

\begin{document}

\title{Incremental Maintenance of Maximal Cliques in a Dynamic Graph}

\author{
Apurba Das \and
Michael Svendsen \and
Srikanta Tirthapura
}

\institute{Apurba Das \at
              Iowa State University \\
              \email{adas@iastate.edu}      
           \and
          	Michael Svendsen \at
              Iowa State University \\
		\email{michael.sven5@gmail.com}
	    \and
		Srikanta Tirthapura \at
		Iowa State University \\
		\email{snt@iastate.edu}
}

\maketitle

\begin{abstract}
We consider the maintenance of the set of all maximal cliques in a dynamic graph that is changing through the addition or deletion of edges. We present nearly tight bounds on the magnitude of change in the set of maximal cliques, as well as the first change-sensitive algorithms for clique maintenance, whose runtime is proportional to the magnitude of the change in the set of maximal cliques. We present experimental results showing these algorithms are efficient in practice, and are faster than prior work by two to three orders of magnitude.
\end{abstract}

\sloppy
\input{intro}

\input{prelim}
\input{size_of_change}
\input{csalgo}
\input{expts}

\section{Conclusion}
We presented change-sensitive algorithms for maintaining the set of maximal cliques in a graph that is changing due to the addition or deletion of edges. We showed nearly tight bounds for the magnitude of change in the set of maximal cliques, due to a change in the set of edges. Our results show that even for the addition of a small number of edges, the change in the number of maximal cliques can be exponential in the size of the graph, in the worst case. Motivated by this, we designed {\em change-sensitive} algorithms, whose time complexity of enumerating the change is proportional to the magnitude of the change. Experimental results show that our algorithms are practical and improve on prior work by orders of magnitude. 

Many interesting research questions remain open, including: (1)~Design of more efficient change-sensitive algorithms for computing $\Lambda^{new}(G,G+H)$, especially for enumerating subsumed cliques. (2)~Computation of the the exact value of $\lambda(n)$, the maximum magnitude of change (3)~Design of change-sensitive algorithms for other dense structures in a graph such as quasi-cliques.



\bibliographystyle{abbrv}
\bibliography{maximal_cliques}


\end{document}

%% file: intro.tex
\section{Introduction}
\label{sec:intro}
Graphs are widely used in modeling linked data, and there has been tremendous interest in efficient methods for finding patterns in graphs, an area often called ``graph mining''. A fundamental task in graph mining is the identification of {\em dense subgraphs}, which are groups of vertices that are tightly interconnected. 

Many applications need to identify dense subgraphs from an evolving graph that is changing with time as new edges are added and old edges are deleted. Examples include real-time identification of stories from Twitter~\cite{AKS+13} through mining dense subgraphs from an evolving graph on entities, and the maintenance of common intervals among genomes~\cite{CRR11} through mining maximal cliques in an appropriately defined dynamic graph. More broadly, identifying dense structures in a graph is applicable to any task that needs to identify and analyze communities with a network, such as the analysis of communities among users in microblogging platforms~\cite{JS+07}, identification of groups of closely linked people in a social network~\cite{hanneman-socialnw,LSZL2011,LSH2008}, identification of web communities~\cite{GKT05,KRRT1999,RH2005}, and even in the construction of the Phylogenetic Tree of Life~\cite{DABMMS2004,SDREL2003,YBE2005}.

Most current methods for identifying dense subgraphs are designed for a static graph. Suppose we used a method designed for a static graph to handle a dynamic graph. If the input graph changes slightly, say, by the addition of a few edges, it is necessary to enumerate all dense subgraphs all over again, even though the set of dense subgraphs may have only changed slightly due to the addition of the new edges. This repeated and redundant work is a source of serious inefficiency, so that methods designed for static graphs are not applicable to a graph that is changing frequently. Different methods are needed, which can handle changes to a graph more efficiently. From a foundational perspective, identifying dense structures in a graph has been a problem of long-standing interest in computer science, but even basic questions remain unanswered on dynamic graphs.

We consider the maintenance of the set of {\em maximal cliques} in a dynamic graph. The maximal clique is perhaps the most fundamental and widely studied dense subgraph. Let $G = (V,E)$ be an undirected unweighted graph on vertex set $V$ and edge set $E$. A clique in $G$ is a set of vertices $C \subseteq V$ such that any two vertices in $C$ are connected to each other in $G$. A clique is called maximal if it is not a proper subset of any other clique.  Let $\cliques(G)$ denote the set of maximal cliques in $G$. Many applications benefit from efficient maintenance of maximal cliques in a dynamic graph, such as described in the work of Chateau et al.~\cite{CRR11} on maintaining common intervals among genomes, Duan et al.~\cite{DL+12} on incremental $k$-clique clustering, Hussain et al.~\cite{HM+15} on maintaining the maximum range-sum query over a point stream.

Suppose that we started from a graph $G = (V, E)$ and the state of the graph changed to $G' = (V, E \cup H)$ through an addition of a set of new edges $H$ to the set of edges in the graph $G$. See Figure~\ref{fig:init} for an example. Let $\Lambda^{new}(G,G') = \cliques(G') \setminus \cliques(G)$ denote the set of maximal cliques that were newly formed in going from $G$ to $G'$, and $\Lambda^{del}(G,G') = \cliques(G) \setminus \cliques(G')$ denote the set of cliques that were maximal in $G$ but are no longer maximal in $G'$. Let $\Lambda(G,G') = \Lambda^{new}(G,G') \cup \Lambda^{del}(G,G')$ denote the symmetric difference of $\cliques(G)$ and $\cliques(G')$. We ask the following questions:

\begin{itemize}
\item How large can the size of $\Lambda(G,G')$ be? To systematically study the problem of maintaining maximal cliques in a dynamic graph, we first need to understand the magnitude of change in the set of maximal cliques.

\item What are efficient methods to compute $\Lambda(G,G')$? Is it possible to have methods that compute $\Lambda(G,G')$ quickly in cases when the size of $\Lambda(G,G')$ is small, and take longer when it is large? Do these methods scale to large graphs?

\end{itemize}

\begin{figure*}
\centering
\begin{tabular}{c}
	\includegraphics[width=.5\textwidth]{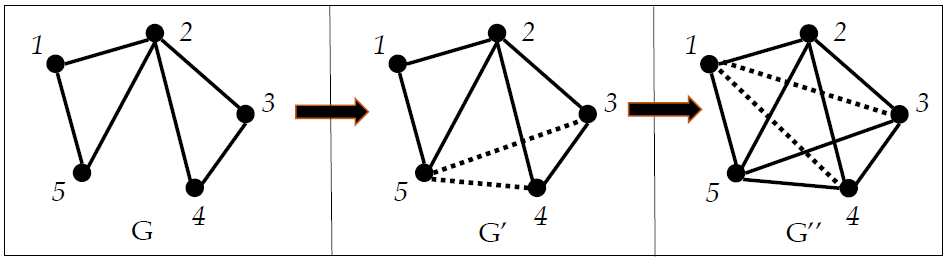}\\
\end{tabular}
\caption{\textbf{Change in maximal cliques due to addition of edges. On the left is the initial graph \textbf{$G$} with maximal cliques $\{1, 2, 5\}$ and $\{2, 3, 4\}$; On the middle is the graph \textbf{$G'$} after adding edges $(3,5)$ and $(4,5)$ to $G$ resulting in new maximal clique $\{2, 3, 4, 5\}$ and only subsumed maximal clique $\{2, 3, 4\}$; On the right is the graph \textbf{$G''$} after adding edges $(1,3)$ and $(1,4)$ to $G'$ resulting in new maximal clique $\{1, 2, 3, 4, 5\}$ and subsumed cliques $\{1, 2, 5\}$ and $\{2, 3, 4, 5\}$.}}
\label{fig:init}
\end{figure*}

\subsection{Contributions}
\noindent \textbf{(A)~Magnitude of Change in the Set of Maximal Cliques:} We present a tight analysis of the magnitude of change in the set of maximal cliques in a graph, when a set of edges are added. When a set of edges $H$ is added to graph $G = (V,E)$ resulting in graph $G' = G \cup H = (V,E \cup H)$. \\

\noindent {\bf (A.1)~} We present nearly matching upper and lower bounds on the maximum size of $\Lambda(G,G \cup H)$, taken across all possible graphs $G$ and edge sets $H$. Let $f(n)$ denote the maximum number of maximal cliques in a graph on $n$ vertices. A result of Moon and Moser~\cite{MM65} shows that $f(n)$ is approximately $3^{n/3}$. We show that by the addition of a small number of edges to the graph $G$ on $n$ vertices, it is possible to cause a change of nearly $2f(n) \approx 2 \cdot 3^{n/3}$. We also note that this is an upper bound on the magnitude of $\Lambda(G,G')$. We present this analysis in Theorem~\ref{thm:lambdan}. \\

\noindent {\bf (A.2)~} We encountered an error in the 50-year old result of Moon and Moser~\cite{MM65} on the number of maximal cliques in a graph, which is directly relevant to our bounds on the change in the set of maximal cliques. We present our correction to their result in Observation~\ref{obs:mm}.

It is easy to see that the set of maximal cliques can change by very little upon the addition of edges. For instance, adding a single edge between two vertices that are part of different components can lead to only a single new maximal clique being added (the clique consisting of a single edge), and no maximal cliques subsumed, so that the total change in the set of maximal cliques is $1$. Thus, we note that the magnitude of the change can vary significantly from one input instance to another.\\

\noindent \textbf{(B)~ Algorithm for Maintaining Maximal Cliques:} We present incremental and decremental algorithms for maintaining the set of maximal cliques of a dynamic graph. We describe result on incremental algorithms here. The results for decremental algorithms are similar.\\

\noindent \textbf{(B.1)~} We present algorithms that take as input $G$ and $H$, and enumerate the elements of $\Lambda(G,G')$ in time proportional to the size of $\Lambda(G,G')$, i.e. the magnitude of the change in the set of maximal cliques. We refer to such algorithms as {\em change-sensitive} algorithms. To our knowledge, these are the first provably change-sensitive algorithms for maintaining the set of maximal cliques in a dynamic graph. The time taken for enumerating newly formed cliques $\Lambda^{new}(G,G')$ is $\mbox{$O(\Delta^3 \rho |\Lambda^{new}(G,G')|)$}$ where $\Delta$ is the maximum degree of a vertex in $G'$ and $\rho$ is the number of edges in $H$. The time taken for enumerating subsumed cliques $\Lambda^{del}(G,G')$ is $O(2^{\rho} |\Lambda^{new}(G,G')|)$. Note that when $\rho$, the size of a batch of edges, is logarithmic in $\Delta$, the cost of enumerating subsumed cliques is of the same order as that of enumerating new cliques.

Our algorithm is based on a careful exploration of a subgraph of $G$ that is local to the set of edges that have been added. Importantly, it does not iterate through existing maximal cliques in the graph. Instead, it directly outputs the maximal cliques that have changed (either added or subsumed). Based on theoretically-efficient algorithms, we present a practical algorithm $\sdiff$ for enumerating new and subsumed cliques, and an efficient implementation. \\

\noindent {\bf (B.2)~} Our methods extend to the decremental case, to handle deletion of edges from the graph. They can also be applied to the fully dynamic case, where the change includes both the addition and deletion of edges from the graph. However, the fully dynamic case is not provably change-sensitive, as discussed in Section~\ref{sec:dec}.\\

\noindent \textbf{(C)~Experimental Evaluation} We present empirical evaluation of our algorithm using real world dynamic graphs as well as synthetic graphs. Our experimental study shows that $\sdiff$ can enumerate change in maximal cliques in a large graph with of the order of a hundred thousand vertices and millions of edges within a few seconds. Our comparison with prior and recent works show that $\sdiff$ significantly outperform prior solutions, including ones due to Stix~\cite{S04}, Ottosen and Vomlel~\cite{OV10}, and Sun et al.~\cite{SW+17}. For example, on the {\tt flickr-growth} graph, our algorithms are faster than~\cite{S04,OV10,SW+17} by a factor of more than a thousand. On the {\tt flickr-growth} graph, in order to maintain the set of maximal cliques over the insertion of 250 batches of 100 edges each, $\sdiff$ took about 40 ms, while prior techniques took anywhere from 5 mins to 2 hrs. Further details are in Section~\ref{sec:expts}.

\subsection{Prior and Related Work}
\label{ref:related}
{\bf Maximal Clique Enumeration in a Static Graph.} There is substantial prior work on enumerating maximal cliques in a static graph, starting from the algorithm based on depth-first-search due to Bron and Kerbosch~\cite{BK73}. A significant improvement to~\cite{BK73} is presented in Tomita et al.~\cite{TTT06}, leading to worst-case optimal time complexity $O(3^{n/3})$ for an $n$ vertex graph~\cite{MM65}. Other work on refinements of~\cite{TTT06,BK73} include~\cite{Koch01}, who presents several strategies for pivot selection to enhance the algorithm in~\cite{BK73}, and a fixed parameter tractable algorithm parameterized by the graph degeneracy~\cite{ELS10,ES11}.

There is a class of algorithms for enumerating structures (such as maximal cliques) in a static graph whose time complexity is proportional to the size of the output -- such algorithms are called ``output-sensitive" algorithms. Many output-sensitive structure enumeration algorithms for static graphs, including~\cite{TI+77,CN85,MU04}, can be seen as instances of a general technique called ``reverse search''~\cite{AF93}. The current best bound on the time complexity of \remove{output-efficient}output-sensitive maximal clique enumeration on a dense graph $G=(V,E)$ is due to~\cite{MU04} which runs with $O(M(n))$ time delay (the interval between outputting two maximal cliques), where $M(n)$ is the time complexity for multiplying two $n\times n$ matrix, which is $O(n^{2.376})$. Further work in this direction includes~\cite{KW+01} and~\cite{JYP88}, who consider the enumeration of maximal independent sets in lexicographic order,~\cite{CK+11}, who consider the external memory model, and~\cite{MXT17}, who consider uncertain graphs. Extensions to parallel frameworks such as MapReduce or MPI are presented in~\cite{SMT14,MT17}.

{\bf Maximal Clique Enumeration in a Dynamic Graph.} In~\cite{S04}, the authors present algorithms for tracking new and subsumed maximal cliques in a dynamic graph when a single edge is added to the graph. These algorithms are not proved to be change-sensitive, even for a single edge. The algorithm due to Stix~\cite{S04} for enumerating new maximal cliques needs to consider (and filter out) maximal cliques in the original graph that remain unaffected due to addition of new edge. This can be wasteful, in terms of update time. Hence, such an algorithm cannot be change-sensitive. For example, consider the case of a graph growing from an empty graph on $10$ vertices to a clique on $10$ vertices. Only one new maximal clique has been formed by this batch, but numerous maximal cliques arise during intermediate steps -- if all these are enumerated, then the time complexity of enumeration is inherently large, even though the magnitude of change is small.

Ottosen and Vomlel~\cite{OV10} present an algorithm to enumerate the change in set of maximal cliques, based on running a maximal clique enumeration algorithm on a smaller graph. Their algorithm supports addition of a set of edges all at once. In contrast with our work, there are no provable performance bounds for this algorithm. Another difference is that the algorithm of~\cite{OV10} may not maintain the exact change in the set of maximal cliques, in certain cases, while our algorithms can maintain the change in the set of maximal cliques exactly.

Sun et al.~\cite{SW+17} present an algorithm for enumerating the change in set of maximal cliques, based on iterating over the set of maximal cliques of the original graph to derive the set of maximal cliques of the updated graph. This need to iterate over currently existing cliques makes the algorithm expensive, especially for cases when the set of maximal cliques does not change significantly due to the update in edge set.
  
Prior algorithms for maximal clique enumeration on a dynamic graph are not proved to be change-sensitive, and do not provide a provable bound on the cost to enumerate the change, or on the magnitude of the change.

{\bf Other Queries on a Dynamic Graph.} Other works on maintaining dense structures on a dynamic graph include methods for the maintenance of $k$-cores~\cite{SG+13,LYM14}, $k$-truss communities~\cite{HC+14}, densest subgraph~\cite{BKV12,MT+15}, and maximal bicliques in a bipartite graph~\cite{DT+17}. Bicliques are complete structures in bipartite graphs, and the theory on the number of substructures and enumeration algorithms is substantially different. The other structures: $k$-core, $k$-truss, and densest subgraph, are different from maximal cliques in that they do not require complete connectivity among different vertices within the structure.

\textbf{Roadmap:} We present preliminaries in Section~\ref{sec:prelim}, followed by bounds on magnitude of change in Section~\ref{sec:sizeofchange}, algorithms for enumerating the change in Section~\ref{sec:csalgo}, and experimental results in Section~\ref{sec:expts}.

%% file: prelim.tex
\section{Preliminaries}
\label{sec:prelim}
\newcommand{\mce}{{\tt MCE}}
We consider a simple undirected graph without self loops or multiple edges. For graph $G$, let $V(G)$ denote the set of vertices in $G$ and $E(G)$ denote the set of edges in $G$. Let $n$ denote the size of $V(G)$, and $m$ denote the size of $E(G)$. For vertex $u \in V(G)$, let $\Gamma_G(u)$ denote the set of vertices adjacent to $u$ in $G$. When the graph $G$ is clear from the context, we use $\Gamma(u)$ to mean $\Gamma_G(u)$. For edge $e=(u,v) \in E(G)$, let $G-e$ denote the graph obtained by deleting $e$ from $E(G)$, but retaining vertices $u$ and $v$ in $V(G)$. Similarly, let $G+e$ denote the graph obtained by adding edge $e$ to $E(G)$. For edge set $H$, let $G+H$ ($G-H$) denote the graph obtained by adding (subtracting) all edges in $H$ to (from) $E(G)$. Let $\Delta(G)$ denote the maximum degree of a vertex in $G$. When the context is clear, we use $\Delta$ to mean $\Delta(G)$. For vertex $v \in V(G)$, let $G-v$ denote the induced subgraph of $G$ on the vertex set $V(G)-\{v\}$, i.e. the graph obtained from $G$ by deleting $v$ and all its incident edges. Let $\cliques_{v}(G)$ denote the set of maximal cliques in $G$ containing $v$.

\textbf{Change-Sensitive Algorithms:} An algorithm for a property $P$ on a dynamic graph is said to be change-sensitive if the time complexity of enumerating the change in $P$ is linear in the magnitude of change (in $P$), and polynomial in the size of the input graph and the size of change in the set of edges. Note that the notion of ``change-sensitive'' for a dynamic graph algorithm is similar to the notion of ``output-sensitive'' in the static graph algorithm where the time complexity is proportional to the size of output times polynomial in other parameters like degree, number of edges etc. For example, see Theorem~\ref{thm:mce}.

An algorithm for a dynamic graph is called {\em incremental} if it can efficiently handle insertion of edges, {\em decremental} if it can handle deletion of edges, and {\em fully dynamic} if it can handle both insertions and deletions. For example, a parallel algorithm due to Simsiri et al.~\cite{ST+16} is an incremental algorithm for graph connectivity, an algorithm due to Thorup~\cite{T99} is a decremental algorithm, and one due to Wulff-Nilsen~\cite{W13} is a fully dynamic algorithm. Our algorithms can be viewed as a change-sensitive incremental algorithm for maximal cliques, and a change-sensitive decremental algorithm for maximal cliques.

\textbf{Results for Static Graphs:} We present some known results about maximal cliques on static graph. Nearly 50 years ago, Moon and Moser~\cite{MM65} considered the question: ``what is the maximum number of maximal cliques that can be present in an undirected graph on $n$ vertices", and gave the following answer. Let $f(n)$ denote the maximum possible number of maximal cliques in a graph on $n$ vertices. A graph on $n$ vertices that achieves $f(n)$ maximal cliques is called a ``Moon-Moser" graph.
\begin{theorem}[Theorem 1, Moon and Moser, \cite{MM65}]
\label{thm:MM65}
\begin{eqnarray*}
f(n) & = & 3^{\frac{n}{3}} \qquad \mbox{~if~ $n \mod 3 = 0$} \\
     & = & 4 \cdot 3^{\frac{n-4}{3}} \qquad \mbox{~if~ $n \mod 3 = 1$} \\ 
     & = & 2 \cdot 3^{\frac{n-2}{3}} \qquad \mbox{~if~ $n \mod 3 = 2$} \\
\end{eqnarray*}
\end{theorem}

We use as a subroutine an output-sensitive algorithm for enumerating all maximal cliques within a (static) graph, using time proportional to the number of maximal cliques. There are multiple such algorithms, for example, due to Tsukiyama et al.~\cite{TI+77}, and due to Makino and Uno~\cite{MU04}. We use the following result due to Chiba and Nishizeki since it provides one of the best possible time complexity bounds for general graphs.  Better results are possible for dense graphs~\cite{MU04} and our algorithm can use other methods as a subroutine also.
\begin{theorem}[Chiba and Nishizeki, \cite{CN85}]
\label{thm:mce}
There is an algorithm $\mce(G)$ that enumerates all maximal cliques in graph $G$ using time $O(\alpha m\mu)$ where $\mu$ is the number of maximal cliques in $G$ and $\alpha$ is the arboricity of $G$\footnote{The arboricity of a graph is no more than the maximum vertex degree of the graph, but could be significantly lesser.} The total space complexity of the algorithm is $O(n+m)$.
\end{theorem}

%% file: size_of_change.tex
\section{Magnitude of Change}
\label{sec:sizeofchange}
\newcommand{\eps}{\varepsilon}

From prior work~\cite{MM65}, the maximum number of maximal cliques in an $n$-vertex graph, denoted by $f(n)$ is known (see Theorem~\ref{thm:MM65}). The result of~\cite{MM65} is relevant for static graphs. In the case of a dynamic graph, a different question is more relevant: {\em what is the maximum change in the set of maximal cliques, that can result from the addition of edges to the graph?} This will give us a bound on the worst case complexity of enumerating the change in the set of maximal cliques.

\subsection{Maximum Possible Change in Maximal Cliques}
We consider the maximum change in the set of maximal cliques upon the addition of edges to the graph. For an integer $n$, let $\lambda(n)$ be the maximum size of $\Lambda(G,G+H)$ taken over all possible $n$ vertex graphs $G$ and edge sets $H$. We present the following result with nearly tight bounds on the value of $\lambda(n)$. Interestingly, our results show that it is possible to change the set of maximal cliques by as much as $\approx 2 \cdot 3^{n/3}$ by the addition of only a few edges to the graph.

\begin{theorem}
\label{thm:lambdan}
\begin{eqnarray*}
\frac{16}{9}f(n) \le & \lambda(n) & < 2f(n) \qquad \mbox{~if~ $(n \mod 3) = 0$} \\
                             & \lambda(n) & = 2f(n)  \qquad \mbox{~if~ $(n \mod 3) = 1$} \\
\frac{11}{6}f(n) \le & \lambda(n) & < 2f(n) \qquad \mbox{~if~ $(n \mod 3) = 2$} \\
\end{eqnarray*}
\end{theorem}

\begin{proof}
We first note that $\lambda(n) \le 2f(n)$ for any integer $n$. To see this, note that for any graph $G$ on $n$ vertices and edge set $H$, it must be true from Theorem~\ref{thm:MM65} that $|\cliques(G)| \le f(n)$ and $|\cliques(G+H)| \le f(n)$. Since $|\Lambda^{new}(G,G+H)| \le |\cliques(G+H)|$ and $|\Lambda^{del}(G,G+H)| \le |\cliques(G)|$, we have $|\Lambda(G,G+H)| = |\Lambda^{new}(G,G+H)| + |\Lambda^{del}(G,G+H)| \le |\cliques(G)| + |\cliques(G+H)| \le 2f(n)$.

The result of Moon and Moser~\cite{MM65} states that for $n \ge 2$, there is only one graph $H_n$ on $n$ vertices (subject to isomorphism) that has $f(n)$ maximal cliques. We show below that there is an error in this result for the case $(n \mod 3) = 1$. Note that the result is still true in the cases $(n \mod 3)$ equals $0$ or $2$. Thus for the cases where $(n \mod 3)$ is $0$ or $2$, adding or deleting edges from $H_n$ leads to a graph with fewer than $f(n)$ maximal cliques, so that we can never achieve a change of $2f(n)$ maximal cliques. Thus we have that for $(n \mod 3)$ equal to $0$ or $2$, $\lambda(n)$ is strictly less than $2f(n)$. The case of $(n \mod 3) = 1$ is discussed separately (see Observation~\ref{obs:mm} below).

We next show that there exists a graph $G$ on $n$ vertices and an edge set $H$ such that the size of $\Lambda(G,G+H)$ is large. See Figure~\ref{fig:size-of-change} for an example. Graph $G$ is constructed on $n$ vertices as follows. Let $\eps > 3$ be an integer. Choose $\eps$ vertices in $V(G)$ into set $V_1$. Let $V_2 = V \setminus V_1$. Edges of $G$ are constructed as follows. 
\begin{itemize}
\itemsep0em 
\item
Each vertex in $V_1$ is connected to each vertex in $V_2$.
\item
Edges are added among vertices of $V_2$ to make the induced subgraph on $V_2$ a Moon-Moser graph on $(n-\eps)$ vertices. Let $G_2$ denote this induced subgraph on $V_2$, which has $f(n-\eps)$ maximal cliques.
\item 
There are no edges among vertices of $V_1$ in $G$.
\end{itemize}

It is clear that for each maximal clique $c$ in $G_2$ and vertex $v \in V_1$, there is a maximal clique in $G$ by adding $v$ to $c$. Thus the number of maximal cliques in $G$ is $|V_1| \cdot |\cliques(G_2)|$. Hence we have 
\begin{equation}
\label{eqn:1}
|\cliques(G)|  = \eps \cdot f(n-\eps)
\end{equation}

To graph $G$, we add the edge set $H$, constructed as follows. $H$ consists of edges connecting vertices in $V_1$, to form a Moon-Moser graph on $\eps$ vertices. Let $G' = G + H$. We note that $\cliques(G)$ and $\cliques(G')$ are disjoint sets. To see this, note that each maximal clique in $G$ contains exactly one vertex from $V_1$, since no two vertices in $V_1$ are connected to each other in $G$. On the other hand, each maximal clique in $G'$ contains more than one vertex from $V_1$, since each vertex $v \in V_1$ is connected to at least one other vertex in $V_1$ in $G'$. Hence, $\Lambda(G,G') = \cliques(G) \cup \cliques(G')$, and
\begin{equation}
\label{eqn:2}
|\Lambda(G,G')| = |\cliques(G)| + |\cliques(G')|
\end{equation}

To compute $|\cliques(G')|$, note that since each vertex in $V_1$ is connected to each vertex in $V_2$, for each maximal clique in $G'(V_1)$ and each maximal clique in $G'(V_2)$, we have a unique maximal clique in $G'$. There are $f(\eps)$ maximal cliques in $G'(V_1)$ and $f(n-\eps)$ maximal cliques in $G'(V_2)$, and hence we have
\begin{equation}
\label{eqn:3}
|\cliques(G')| = f(\eps) \cdot f(n-\eps)
\end{equation}

Putting together Equations~\ref{eqn:1},~\ref{eqn:2}, and~\ref{eqn:3} we get 

\begin{equation}
\label{eqn:4}
|\Lambda(G,G')| = (\eps + f(\eps)) \cdot f(n-\eps)
\end{equation}

Let $F(\eps) = (\eps + f(\eps))f(n-\eps)$. We compute the value of $\eps (> 3)$ at which $F(\eps)$ is maximized. To do this, we consider three different cases depending on the value of $(n \mod 3)$, and omit the calculations. If $n\mod 3 = 0$, $F(\eps)$ is maximized at $\eps = 4$ and the maximum value $F(4) = \frac{16}{9}f(n)$. If $n\mod 3 = 1$, $F(\eps)$ is maximized at $\eps = 4$ and $F(4) = 2f(n)$. And finally if $n\mod 3 = 2$, $F(\eps)$ is maximized at $\eps= 5$ and $F(5) = \frac{11}{6}f(n)$. This completes the proof.
\qed \end{proof}

\begin{figure*}
\centering
\begin{tabular}{c}
	\includegraphics[width=.5\textwidth]{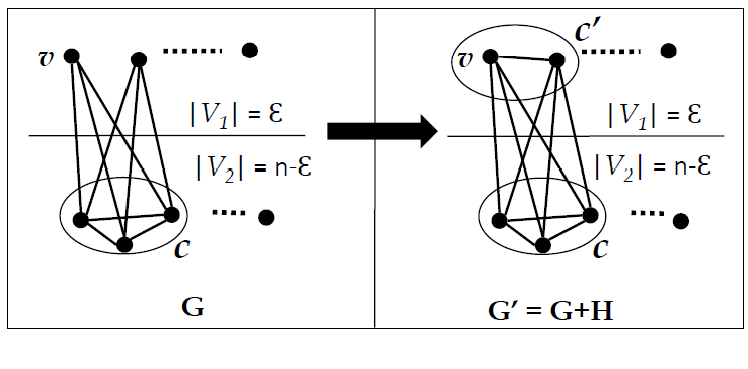}\\
\end{tabular}
\caption{\textbf{Construction showing a large change in set of maximal cliques when a few edges are added. $V_1$ is the set of vertices above the horizontal line and $V_2$ is the set of vertices below the horizontal line where $V_1\cup V_2 = V$ and $V_1\cap V_2 = \phi$. On the left is $G$, the original graph with $n$ vertices where each vertex in $V_1$ is connected to each vertex in $V_2$, and $V_1$ is an independent set. In $G$, the induced subgraph $G_2$ on vertex set $V_2$ forms a {\tt Moon-Moser} graph. On the right is $G'$, the graph formed after adding edge set $H$ to $G$ such that the induced subgraph on vertex set $V_1$ becomes a {\tt Moon-Moser} graph. Let $c$ be a clique in $G_2$, and $c'$ a new clique in $G'$ formed among vertices in $V_1$. Note that $c\cup\{v\}$ was a maximal clique in $G$, and is now subsumed by a new maximal clique $c\cup c'$.}}
\label{fig:size-of-change}
\end{figure*}

\begin{figure*}
\centering
\begin{tabular}{c}
\includegraphics[width=.5\textwidth]{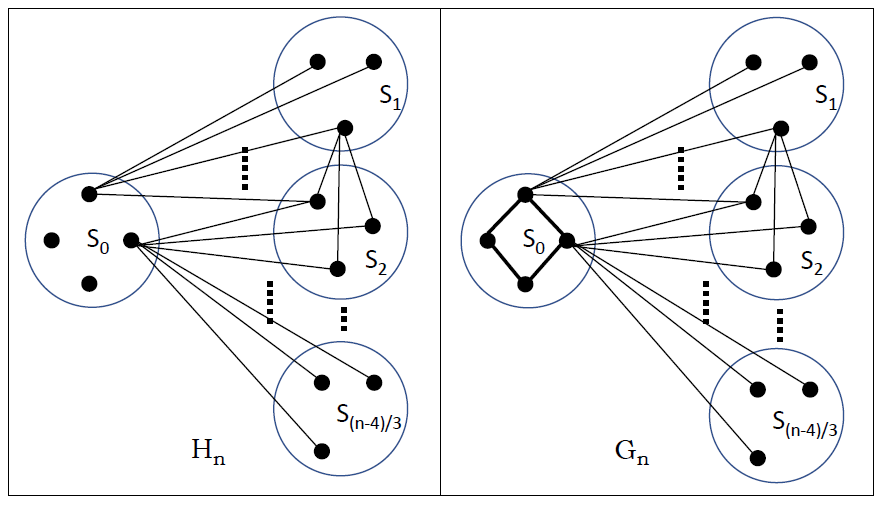}\\
\end{tabular}
\caption{\textbf{On the left is $H_n$ where each vertex $v$ in $S_i$ is connected to each vertex $u$ in $S_j$, $i \neq j$. On the right is $G_n$ which is formed from $H_n$ by adding four edges to $S_0$. For the case $(n \mod 3) = 1$, $H_n$ and $G_n$ are non-isomorphic graphs on $n$ vertices, with $f(n)$ maximal cliques each, showing a counterexample to Theorem 2 of Moon and Moser~\protect\cite{MM65}.}}
\label{fig:moon-moser}
\end{figure*}

\subsection{An Error in a Result of Moon and Moser (1965)}
Moon and Moser~\cite{MM65}, in Theorem 2 in their paper, claim ``For any $n \ge 2$, if a graph $G$ has $n$ nodes and $f(n)$ cliques, then $G$ must be equal to $H_n$", where $H_n$ is a specific graph, described below. We found that this theorem is incorrect for the case when $(n \mod 3) = 1$.

The error is as follows (see Figure~\ref{fig:moon-moser}). For $(n \mod 3) = 1$, the graph $H_n$ is constructed on vertex set $V_n=\{1,2,\ldots,n\}$ by taking vertices $\{1,2,3,4\}$ into a set $S_0$ and dividing the remaining vertices into groups of three, as sets $S_1,S_2,\ldots,S_{\frac{n-4}{3}}$. In graph $H_n$, edges are added between any two vertices $u,v$ such that $u \in S_i$, $v \in S_j$ and $i \neq j$. This graph $H_n$ has $4\cdot 3^{\frac{n-4}{3}}$ maximal cliques, since we can make a maximal clique by choosing a vertex from $S_0$ (4 ways), and one vertex from each $S_i, i > 0$ (3 ways for each such $S_i, i = 1\ldots (n-4)/3$). 

Contradicting Theorem 2 in~\cite{MM65}, we show there is another graph $G_n$ that is different from $H_n$, but still has the same number of maximal cliques. $G_n$ is the same as $H_n$, except that the vertices within $S_0$ are connected by a cycle of length $4$. In this case, we can still construct $4\cdot 3^{\frac{n-4}{3}}$ maximal cliques, since we can make a maximal clique by choosing two connected vertices in $S_0$ (4 ways to do this), and one vertex from each $S_i, i > 0$ (3 ways for each such $S_i, i = 1\ldots (n-4)/3$).

\begin{observation}
\label{obs:mm}
For the case $(n \mod 3) = 1$, there are two distinct non-isomorphic graphs $H_n$ and $G_n$  described above, that have $4\cdot 3^{\frac{n-4}{3}}$ maximal cliques, which is the maximum possible. This is a correction to Theorem 2 of Moon and Moser~\cite{MM65}, which states that there is only one such graph, $H_n$.
\end{observation}

This observation enables us to have $\lambda(n) = 2f(n)$ for the case $(n \mod 3) = 1$. By starting with graph $H_n$ and by adding edges to make it $G_n$, we remove $f(n)$ maximal cliques and introduce $f(n)$ maximal cliques, leading to a total change of $2f(n)$.

%% file: csalgo.tex
\section{Enumeration of Change in Set of Maximal Cliques}
\label{sec:csalgo}
In this section we present algorithms for enumerating the change in the set of maximal cliques. In Section~\ref{sec:newc}, we first present an algorithm with provable theoretical properties for enumerating new maximal cliques that arise due to the addition of a batch of edges, followed by an algorithm with good practical performance in Section~\ref{sec:newc-prac}. In Section~\ref{sec:subc}, we present an algorithm for enumerating subsumed cliques due to the addition of new edges. We then consider the decremental case where edges are deleted from the graph in Section~\ref{sec:dec}. 
For graph $G$ and edge set $H$, when the context is clear, we use $\Lambda^{new}$ to mean $\Lambda^{new}(G,G+H)$ and similarly $\Lambda^{del}$ to mean $\Lambda^{del}(G,G+H)$.


\subsection{Enumeration of New Maximal Cliques When Edge Set $H$ is Added}
\label{sec:newc}
When a set of edges $H$ is added to the graph $G$, let $G'$ denote the graph $G+H$. One approach to enumerating new cliques in $G'$ is to simply enumerate all cliques in $G'$ using an output-sensitive algorithm such as~\cite{CN85}, suppress those cliques that were also present in $G$, and output the rest. However, this is clearly not change-sensitive, since it is possible that most cliques in $G'$ will not be finally output. An important point to note here is that most similar approaches, that involve enumerating maximal cliques in a certain graph, followed by suppressing cliques that do not belong to $\Lambda^{new}$, run this risk of sometimes having to suppress most of the cliques that were enumerated, and such approaches will not be change-sensitive. In the following, we present a simple approach, which, at its core, directly outputs cliques from $\Lambda^{new}$, and does not output cliques that do not belong to $\Lambda^{new}$. 

For edge $e \in H$, let $C'(e)$ denote the set of maximal cliques in $G'$ that contain edge $e$. We first present the following observation that $\Lambda^{new}$, the set of all new maximal cliques, is precisely the set of all maximal cliques in $G'$ that contain at least one edge from $H$.

\begin{lemma}
\label{lem:nc1}
\[
\Lambda^{new}(G,G') = \cup_{e \in H} C'(e)
\]
\end{lemma}

\begin{proof}
We first note that each clique in $\Lambda^{new}$ must contain at least one edge from $H$. We use proof by contradiction. Consider a clique $c \in \Lambda^{new}$. If $c$ does not contain an edge from $H$, then $c$ is also a clique in $G$, and hence cannot belong to $\Lambda^{new}$. Hence, $c \in C'(e)$ for some edge $e \in H$, and $c \in  \cup_{e \in H} C'(e)$. This shows that $\Lambda^{new} \subseteq \cup_{e \in H} C'(e)$. Next consider a clique $c \in \cup_{e \in H} C'(e)$. It must be the case that $c \in C'(h)$ for some $h$ in $H$. Thus $c$ is a maximal clique in $G'$. Since $c$ contains edge $h \in H$, $c$ cannot be a clique in $G$. Thus $c \in \Lambda^{new}$. This shows that $\cup_{e \in H} C'(e) \subseteq \Lambda^{new}$.
\qed\end{proof}

We now consider efficient ways of enumerating cliques from $\cup_{e \in H} C'(e)$. For an edge $e \in H$, the enumeration of cliques in $C'(e)$ is reduced to the enumeration of {\em all} maximal cliques in a specific subgraph of $G'$, as follows. Let $u$ and $v$ denote the endpoints of $e$, and let $G'_e$ denote the induced subgraph of $G'$ on the vertex set $\{u,v\}\cup\{\Gamma_{G'}(u) \cap \Gamma_{G'}(v)\}$ i.e. the set of vertices adjacent to both $u$ and $v$ in $G'$, in addition to $u$ and $v$. For example, see Figure~\ref{fig:csnew} for  construction of $G'_e$. 

%

\begin{figure*}
\centering
\begin{tabular}{c}
	\includegraphics[width=.5\textwidth]{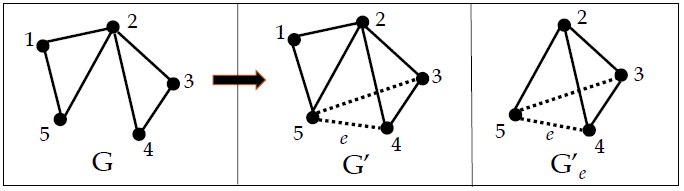}\\
\end{tabular}
\caption{\textbf{Illustration of Lemma~\ref{lem:nc2} that, the set of new maximal cliques in $G'$ containing $e = (4,5)$, i.e. the single clique $\{2,3,4,5\}$, is exactly the set of all maximal cliques in $G'_e$.}} 
\label{fig:csnew}
\end{figure*}
\begin{lemma}
\label{lem:nc2}
\[
For~each~e\in H,~C'(e) = \cliques(G'_e) 
\] 
\end{lemma}
\begin{proof}
First we show that $C'(e)\subseteq\cliques(G'_e)$.
Consider a clique $c$ in $C'(e)$, i.e. a maximal clique in $G'=G+H$ containing edge $e$. Hence $c$ must contain both $u$ and $v$. Every vertex in $c$ (other than $u$ and $v$) must be connected to both $u$ and to $v$ in $G'$, and hence must be in $\Gamma_{G'}(u) \cap \Gamma_{G'}(v)$. Hence $c$ must be a clique in $G'_e$. Since $c$ is a maximal clique in $G'$, and $G'_e$ is a subgraph of $G'$, $c$ must also be a maximal clique in $G'_e$. Hence we have that $c \in \cliques(G'_e)$, leading to $C'(e) \subseteq \cliques(G'_e)$.

Next, we show that $\cliques(G'_e) \subseteq C'(e)$. Consider any maximal clique $d$ in $G'_e$. We note the following in $G'_e$: (1)~every vertex in $G'_e$ (other than $u$ and $v$) is connected to $u$ as well as $v$ (2)~$u$ and $v$ are connected to each other. Due to these conditions, $d$ must contain both $u$ and $v$, and hence also edge $e=(u,v)$. Clearly, $d$ is a clique in $G'$ that contains edge $e$. We now show that $d$ is a maximal clique in $G'$. Suppose not, and we could add vertex $v'$ to $d$ and it remained a clique in $G'$. Then, $v'$ must be in $\Gamma_{G'}(u) \cap \Gamma_{G'}(v)$, and hence $v'$ must be in $G'_e$, so that $d$ is not a maximal clique in $G'_e$, which is a contradiction. Hence, it must be that $d$ is a maximal clique in $G'$ that contains edge $e$, and $d \in C'(e)$.
\qed\end{proof}
Following Lemma~\ref{lem:nc2}, in Figure~\ref{fig:csnew}, $\{2,3,4,5\}$ is a new maximal clique in $G'$ that contains $e = (4,5) \in H,  H = \{(3,5), (4,5)\}$. Note that $\{2,3,4,5\}$ is also a maximal clique in $G'_e$.

Our change-sensitive algorithm, $\csnew$ (Algorithm~\ref{algo:newc1}) is based on the above observation, and uses an output-sensitive algorithm \mce, due to~\cite{CN85}, to enumerate all maximal cliques in $G'_e$.
\begin{algorithm}
\DontPrintSemicolon
\caption{$\csnew(G,H$)}
\label{algo:newc1}
\KwIn{$G$ - Input graph, $H$ - Set of $\rho$ edges added to $G$}
\KwOut{All cliques in $\Lambda^{new}$, each clique output once}
Consider edges of $H$ in an arbitrary order $e_1, e_2,\ldots,e_{\rho}$\;
$G' \gets G + H$\;
\For{$i = 1\ldots \rho$}{
	$e \gets e_i = (u,v)$\;
	$V_e \gets \{u,v\}\cup\{\Gamma_{G'}(u)\cap\Gamma_{G'}(v)\}$\;
	$G'_e \gets$ graph induced by $V_e$ on $G'$\;
	Generate cliques using $\mce(G'_e)$. For each clique $c$ thus generated, output $c$ only if $c$ does not contain an edge $e_j$ for $j < i$\;
}
\end{algorithm}

\begin{theorem}
\label{thm:cs-newc}
$\csnew$ enumerates the set of all new cliques arising from the addition of $H$ in time $O(\Delta^3 \rho |\Lambda^{new}|)$ where $\Delta$ is the maximum degree of a vertex in $G'$. The space complexity is $O(|E(G+H)| + |V(G+H)|)$.
\end{theorem}

\begin{proof}
We first prove the correctness of the algorithm. From Lemmas~\ref{lem:nc1} and~\ref{lem:nc2}, we have that by enumerating $\cliques(G'_e)$ for every $e \in H$, we enumerate $\Lambda^{new}$. Our algorithm does exactly that, and enumerates $\cliques(G'_e)$ using Algorithm $\mce$. Note that each clique $c \in \Lambda^{new}$ is output exactly once though $c$ maybe in $\cliques(G'_e)$ for multiple edges $e \in H$. This is because $c$ is output only for edge $e$ that occurs earliest in the pre-determined ordering of edges in $H$. 

For the runtime, consider that the algorithm iterates over the edges in $H$. In an iteration involving edge $e$, it constructs a graph $G'_e$ and runs $\mce(G'_e)$. Note that the number of vertices in $G'_e$ is no more than $\Delta+1$, and is typically much smaller, since it is the size of the intersection of two vertex neighborhoods in $G'$. Since the arboricity of a graph is less than its maximum degree, $\alpha'\le \Delta$ where $\alpha'$ is the arboricity of $G'_e$. Further, the number of edges in $G'_e$ is {$O(\Delta^2)$}. The set of maximal cliques generated in each iteration is a subset of $\Lambda^{new}$, hence the number of maximal cliques generated from each iteration is no more than $|\Lambda^{new}|$. Applying Theorem~\ref{thm:mce}, we have that the runtime of each iteration is $O(\Delta^3 |\Lambda^{new}|)$. Since there are $\rho$ iterations, the result on runtime follows.

For the space complexity, we note that the algorithm does not store the set of new cliques in memory at any point. The space required to construct $G'_e$ is linear in the size of $G'=(G+H)$, and so is the space requirement of Algorithm $\mce(G'_e)$, from Theorem~\ref{thm:mce}. Hence the total space requirement is linear in the number of edges in $G+H$.
\qed\end{proof}

\subsection{Practical Algorithm for Enumerating New Maximal Cliques}
\label{sec:newc-prac}
The algorithm $\csnew$ uses as a subroutine Algorithm $\mce$ (Chiba and Nishizeki~\cite{CN85}) to enumerate maximal cliques within a subgraph of $G$. While $\mce$ has theoretical properties of being output-sensitive, in practice, it is not the fastest algorithm for maximal clique enumeration. In practice, the most efficient algorithms for maximal clique enumeration in a static graph are based on depth-first search using a technique called ``pivoting", such as the algorithm due to Tomita et al.~\cite{TTT06}.  While the runtime of these algorithms are not provably output-sensitive, they are faster in practice than those algorithms that are provably output-sensitive.

The algorithm due to Tomita et al.~\cite{TTT06}, which we call $\tomita$, is a recursive algorithm based on backtracking for enumerating $\cliques(G)$, given $G$. In its recursive procedure, it maintains a currently found clique, not necessarily maximal, and adds vertices one to the current clique, declaring the current clique to be maximal when no further vertices can be added. The vertices are considered in a carefully chosen order using a method called ``pivoting". $\tomita$ is shown to be worst-case optimal with a runtime of $O(3^{n/3})$ for an $n$ vertex graph. 
It is possible to improve the performance of the $\csnew$ algorithm by directly using $\tomita$ in place of $\mce$. In the following, we show how to do even better. 

\paragraph*{Reducing Redundant Clique Computation: } 
Note that $\csnew(G,H)$ may compute the same clique $c$ multiple times, for example, if $c \in C'(e_1)$ and $c \in C'(e_2)$ for $e_1 \neq e_2$. Duplicates are suppressed prior to emitting the cliques, by outputting $c$ only for one of the edges among $\{e_1,e_2\}$, but the algorithm still pays the computational cost of computing a clique such as $c$ multiple times.  We present Algorithm $\csnewttt$ that eliminates the cost of having to even compute a clique such as $c$ multiple times.

\begin{algorithm}
\DontPrintSemicolon
\caption{$\tomitaE(\mathcal{G},K,\cand,\fini, \mathcal{E})$}
\label{algo:tomitaE}
\KwIn{$\mathcal{G}$ - The input graph, $K$ - a non-maximal clique to extend \\ 
$\cand$ - Set of vertices that may extend $K$,  
$\fini$ - vertices that have been used to extend $K$ \\ 
$\mathcal{E}$ - set of edges to exclude}
\If{$(\cand = \emptyset)$ \& $(\fini = \emptyset)$}{
	Output $K$ and return\;
}
$\pivot \gets (u \in \cand \cup \fini)$ such that $u$ maximizes the size of $\cand  \cap \Gamma_{\mathcal{G}}(u)$\;
$\ext \gets \cand - \Gamma_{\mathcal{G}}(\pivot)$\;
\For{$q \in$ \ext}{
	$K_q \gets K\cup\{q\}$\;
	\If{$K_q \cap \mathcal{E} \neq \emptyset$}{
		$\cand \gets \cand - \{q\}$ ; $\fini \gets \fini\cup\{q\}$\;
		continue\;
	}
	$\cand_q \gets \cand\cap\Gamma_{\mathcal{G}}(q)$ ; $\fini_q \gets \fini\cap\Gamma_{\mathcal{G}}(q)$\;
	$\tomitaE(\mathcal{G},K_q,\cand_q,\fini_q,\mathcal{E})$\;
	$\cand \gets \cand-\{q\}$ ; $\fini \gets \fini \cup\{q\}$\;
}
\end{algorithm}

Algorithm $\csnewttt$ uses as a subroutine Algorithm $\tomitaE$, an extension of the $\tomita$ algorithm, which enumerates all maximal cliques of an input graph that avoid a given set of edges. While $\tomita$ simply takes a graph as input and enumerates all maximal cliques within the graph, $\tomitaE$ takes an additional input, a set of edges $\mathcal{E}$, and only enumerates those cliques within the graph that do not contain any edge from $\mathcal{E}$. We present the recursive version of $\tomitaE$, which takes as input five parameters -- an input graph $G$, three sets of vertices $K$, $\cand$, and $\fini$, and a set of edges $\mathcal{E}$. The algorithm outputs every maximal clique in $G$ that (a)~contain all vertices in $K$, (b)~zero or more vertices in $\cand$, (c)~none of the vertices in $\fini$, and (d)~none of the edges in $\mathcal{E}$.  

A description of $\tomitaE$ is presented in Algorithm~\ref{algo:tomitaE}. We note that this algorithm follows the structure of the recursion in the $\tomita$ algorithm, and incorporates additional pruning of search paths, by avoiding paths that contain an edge from $\mathcal{E}$. In particular, in Line~7 of $\tomitaE$, if the clique $K_q$ (formed after adding vertex $q$ to $K$) contains an edge from $\mathcal{E}$, then the rest of the search path, which will continue adding more vertices, is not explored further. Instead the algorithm backtracks and tries to extend the clique $K$ by adding other vertices. 


\begin{figure*}
\centering
\begin{tabular}{c}
	\includegraphics[width=.8\textwidth]{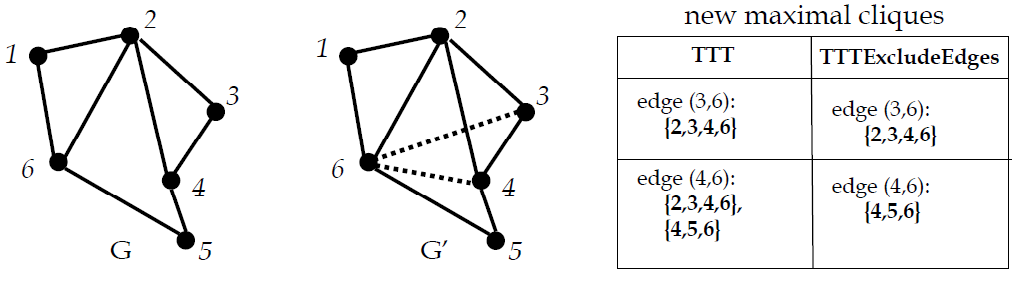}\\
\end{tabular}
\caption{\textbf{Enumeration of new maximal cliques when graph changes from $G$ to $G'$ due to addition of new edges $(3,6)$ and $(4,6)$. Consider the ordering of edges as $(3,6)$ followed by $(4,6)$. There are two new maximal cliques containing edge $(4,6)$, which are $\{4,5,6\}$ and $\{2,3,4,6\}$. With $\tomitaE$,  only $\{4,5,6\}$ is enumerated when considering edge $(4,6)$, since $\{2,3,4,6\}$ has already been enumerated while considering edge $(3,6)$.}}
\label{fig:tttext}
\end{figure*}

\begin{algorithm}[t]
\DontPrintSemicolon
\caption{$\csnewttt(G,H)$}
\label{algo:newc-ttt}
\KwIn{$G$ - input graph \\ \hspace{1cm} $H$ - Set of $\rho$ edges being added to $G$}
\KwOut{Cliques in $\Lambda^{new} = \cliques(G+H)\setminus\cliques(G)$}
$G' \gets G + H$ ; $\mathcal{E}\gets\phi$\;
Consider edges of $H$ in an arbitrary order $e_1, e_2,\ldots,e_{\rho}$\;

\For{$i \gets 1,2, \ldots,\rho$}{
	$e \gets e_i = (u,v)$\;
	
	$V_e \gets \{u,v\}\cup\{\Gamma_{G'}(u)\cap\Gamma_{G'}(v)\}$\;
	$\mathcal{G} \gets$ Graph induced by $V_e$ on $G'$\;
	$K \gets \{u,v\}$\;
	$\cand \gets V_e \setminus \{u,v\}$ ; $\fini \gets \emptyset$\;
	$S\gets\tomitaE(\mathcal{G}, K, \cand, \fini, \mathcal{E})$\;
	$\Lambda^{new}\gets\Lambda^{new}\cup S$\;
	$\mathcal{E}\gets\mathcal{E}\cup e_i$\;
}
\end{algorithm}

Our algorithm for enumerating new maximal cliques $\csnewttt$ (Algorithm~\ref{algo:newc-ttt}) is an adaptation of change-sensitive algorithm $\csnew$ (Algorithm~\ref{algo:newc1}) where we use $\tomitaE$ instead of the output-sensitive $\mce$. In particular, while enumerating all new cliques containing edge $e_i$, $\csnewttt$ enumerates only those cliques that exclude edges $\{e_1,e_2\ldots,e_{i-1}\}$. Note that in $\csnewttt$, there is no further duplicate suppression required, since the call to $\tomitaE$ does not return any cliques that contain an edge from $\mathcal{E}$. This is more efficient than first enumerating duplicate cliques, followed by suppressing duplicates before emitting them. This idea makes $\csnewttt$ more efficient in practice than $\csnew$.


The correctness of $\csnewttt$ follows in a similar fashion to that of Algorithm~$\csnew$ proved in Theorem~\ref{thm:cs-newc}, except that we also need a proof of the guarantee provided by Algorithm~$\tomitaE$, which we prove below.
\begin{lemma}
\label{lem:tomitaE}
$\tomitaE(\mathcal{G},K,\cand,\fini, \mathcal{E})$ (Algorithm~\ref{algo:tomitaE}) returns all maximal cliques $c$ in $\mathcal{G}$ such that (1)~$c$ contains all vertices from $K$, (2)~remaining vertices in $c$ are chosen from $\cand$, (3)~$c$ contains no vertex from $\fini$ and (4)~$c$ does not contain any edges in $\mathcal{E}$.
\end{lemma}

\begin{proof}
We note that $\tomitaE$ matches the original $\tomita$ algorithm, except for lines $7$ to $9$. Hence, if we do not consider lines $7$ to $9$ in $\tomitaE$, the algorithm becomes $\tomita$, and by the correctness of $\tomita$ (Theorem 1~\cite{TTT06}), all maximal cliques $c$ in $\mathcal{G}$ are returned. 

Now consider lines $7$ to $9$ in $\tomitaE$. Clearly, (1), (2), (3) are preserved for each maximal clique $c$ generated by $\tomitaE$. Now to complete the correctness proof of $\tomitaE$, along with proving (4), we also need to prove that each maximal clique $c$ in $\mathcal{G}$ that does not contain any edge in $\mathcal{E}$ is generated by $\tomitaE$. 

Assume there exists a maximal clique $c$ in $\mathcal{G}$, which contains an edge in $\mathcal{E}$, which is output by the $\tomitaE$ algorithm, assume the offending edge is $e = (q,v)$. Suppose that vertex $v$ was added to our expanding clique first. Then, as $q$ is processed, line $7$ of the algorithm will return back true as $e \in K_q$ and $\mathcal{E}$, thus $q$ will not be added to the clique, and $c$ will not be reported as maximal, a contradiction.

Now we will show that, if a maximal clique does not contain an edge from $\mathcal{E}$, the clique will be generated. Consider a maximal clique $c$ in $\mathcal{G}$ that contains no edge from $\mathcal{E}$ but $c$ is not generated by $\tomitaE$. The only reason for $c$ not being generated is the inclusion of lines $7$ to $9$ (of $\tomitaE$) to $\tomita$ resulting in $\tomitaE$, because, otherwise, $c$ would be generated due to correctness of $\tomita$. So in $\tomitaE$, during the expansion of $K$ towards $c$, there exists a vertex $q \in c$ such that line $7$ in $\tomitaE$ is satisfied and $c$ never gets a chance to be generated as $q$ is excluded from $\cand$ and included in $\fini$ (line $8$). This implies that $c$ contains at least an edge in $\mathcal{E}$, because otherwise, condition at line $7$ would never be satisfied. This is a contradiction. Hence,  $c$ must be generated. This completes the proof.
\qed\end{proof}

\subsection{Enumeration of Subsumed Maximal Cliques}
\label{sec:subc}
We now consider the enumeration of subsumed cliques, i.e. the set $\cliques(G) \setminus \cliques(G+H)$. A subsumed clique $c'$ still exists in $G'=G+H$, but is now a part of a larger clique in $G'$. Such a larger clique must be a part of $\Lambda^{new}$. Thus, an algorithm idea is to check each new clique $c$ in $\Lambda^{new}$ to see if $c$ subsumed any maximal clique $c'$ in $G$. In order to see which maximal cliques $c'$ may have subsumed, we note that any maximal clique subsumed by $c$ must also be a maximal clique within  subgraph $c-H$. Thus, one approach is to enumerate all maximal cliques in $c-H$ and for each such generated clique $c'$, we check whether $c'$ is maximal in $G$ by verifying maximality of $c'$ in $G$. This algorithm can be implemented in space proportional to the size of $G+H$, since it can directly use an algorithm for maximal clique enumeration such as $\mce$.

However, in practice, checking each potential clique for maximality is a costly operation since it potentially needs to consider the neighborhood of every vertex of the clique. An alternative approach to avoid this costly maximality check is to store the set of maximal cliques $\cliques(G)$ and check if $c'$ is in $\cliques(G)$. The downside of this approach is that the space required to store the clique set can be high. 

Hence, we considered another approach to subsumed cliques, where we reduce the memory cost by storing signatures of maximal cliques as opposed to the cliques themselves. The signature is computed by representing a clique in a canonical fashion (for instance, representing the clique as a list of vertices sorted by their ids.) as a string followed by computing a hash of this string. By storing only the signatures and not the cliques themselves, we are able to check if a clique is a current maximal clique, and at the same time, pay far lesser cost in memory when compared with storing the clique itself. The algorithm is described in Algorithm~\ref{algo:cs-sub1}, With this approach of storing signatures instead of storing the cliques themselves, there is a (small) chance of collision of signatures, which means for two different cliques $C_1$ and $C_2$ the signature might be the same. This might result in false positives meaning that some cliques might wrongly be concluded as subsumed cliques. However, the probability of the event that the hash values of two different cliques are same is extremely low with the use of a hashing algorithm such as 64-bit murmur hash~\footnote{\url{https://sites.google.com/site/murmurhash/}}. In our experiments, we observed that the set of subsumed cliques reported with the use of signature is always the same as the actual set of subsumed cliques. If it is extremely important to avoid false positives, we can explicitly check a potential subsumed clique for maximality in the original graph.

In Algorithm~\ref{algo:cs-sub1}, Lines $4$ to $12$ describes the procedure for computing $S$, the set of all maximal cliques in $c - H$ and Lines $13$ to $15$ decide which among the maximal cliques in $S$ are subsumed. For computing maximal cliques in $c - H$, we only consider the edges in $H$ that are present in $c$ as we can see in Line $4$. We prove that $S$ is the set of all maximal cliques in $c - H$ in the following lemma:
\begin{lemma}
\label{lem:cliques_in_cminusH}
In Algorithm~\ref{algo:cs-sub1}, for each $c \in \Lambda^{new}$, $S$ contains all maximal cliques in $c - H$.  
\end{lemma}

\begin{proof}
First note that, we only consider the set of all edges $H_1\subseteq H$ which are present in $c$ (line $4$). Clearly computing maximal cliques in $c - H$ reduces to computing maximal cliques in $c - H_1$. We prove this using induction on $k$, the number of edges in $H_1$. Suppose $k = 1$ so that $H_1$ is a single edge, say $e_1 = \{u,v\}$. Then clearly, $c - H_1$ has two maximal cliques, $c \setminus \{u\}$ and $c \setminus \{v\}$, proving the base case. Suppose that for any set $H_1$ of size $k$, it is true that all maximal cliques in $c - H_1$ have been generated using induction hypothesis. Consider a set $H'_1 = \{e_1,e_2,...,e_{k+1}\}$ with $(k+1)$ edges. Now each maximal clique $c'$ in $c - H_1$ either remains a maximal clique within $c - H'_1$ (if at least one endpoint of $e_{k+1}$ is not in $c'$), or leads to two maximal cliques in $c - H'_1$ (if both endpoints of $e_{k+1}$ are in $c'$). Thus lines $4$ to $12$ in Algorithm~\ref{algo:cs-sub1} generates all maximal cliques in $c - H$.
\qed\end{proof}

\remove{
\begin{algorithm}
\DontPrintSemicolon
\caption{$\cssub_1(G,H,\Lambda^{new})$}
\label{algo:cs-sub1}
\KwIn{$G$ - Input Graph \\ \hspace{1cm} $H$ - Edge set being added to $G$ \\ \hspace{1cm} $\Lambda^{new}$ - set of new maximal cliques in \st{$G\cup H$} {\color{blue}$G + H$}}
\KwOut{All cliques in $\Lambda^{del} = \cliques(G) - \cliques(G+H)$}
\For{$c \in \Lambda^{new}$}{
	$c^{-H} \gets c\setminus H$\;
	{\color{red}$C_{check} \gets \tomita(c^{-H})$}\;
	\For{$c' \in C_{check}$}{
		\If{$c'$ is maximal in $G$}{
			Output $c'$ into $\Lambda^{del}$\;
		}
	}
}
\end{algorithm}

\begin{algorithm}[t]
\DontPrintSemicolon
\caption{$\cssub_1(G,H,C,\Lambda^{new})$}
\label{algo:cs-sub1}
\KwIn{$G$ - Input Graph \\ \hspace{1cm} $H$ - Edge set being added to $G$ \\ \hspace{1cm} $C$ - Set of maximal cliques in $G$ \\ \hspace{1cm} $\Lambda^{new}$ - set of new maximal cliques in $G\cup H$}
\KwOut{All cliques in $\Lambda^{del} = \cliques(G) - \cliques(G+H)$}
$\Lambda^{del} \gets \emptyset$\;
\For{$c \in \Lambda^{new}$}{
	$c^{-H} \gets c\setminus H$\;
	\st{$C_{check} \gets \tomita(c^{-H}, K=\emptyset, \cand=c^{-H}, \fini=\emptyset)$}\;
	{\color{blue}$C_{check} \gets \mce(c^{-H})$}\;
	\For{$c' \in C_{check}$}{
		\If{$c' \in C$}{
			$\Lambda^{del} \gets \Lambda^{del} \cup c'$\;
		}
	}
}
\end{algorithm}
}

\begin{algorithm}[t]
\DontPrintSemicolon
\caption{$\cssub(G,H,C,\Lambda^{new})$}
\label{algo:cs-sub1}
\KwIn{$G$ - Input Graph \\ \hspace{1cm} $H$ - Edge set being added to $G$ \\ \hspace{1cm} $C$ - Set of maximal cliques in $G$ \\ \hspace{1cm} $\Lambda^{new}$ - set of new maximal cliques in $G + H$}
\KwOut{All cliques in $\Lambda^{del} = \cliques(G) \setminus \cliques(G+H)$}
$\Lambda^{del} \gets \emptyset$\;
\For{$c \in \Lambda^{new}$}{
		$S\gets \{c\}$\;
		\For{$e = (u,v) \in E(c) \cap H$}{
			$S'\gets \phi$\;
			\For{$c' \in S$}{
				\If{$e \in E(c')$}{
					$c_1 = c'\setminus\{u\}$ ; $c_2 = c'\setminus\{v\}$\;
					$S' \gets S'\cup c_1$ ; $S' \gets S'\cup c_2$\;
					
				}
				\Else{
					
					$S' \gets S'\cup c'$\;
				}
			}
			$S\gets S'$\;
		}
%
	\For{$c' \in S$}{
		\If{$c' \in C$}{
			$\Lambda^{del} \gets \Lambda^{del} \cup c'$\;
			$C \gets C\setminus c'$\;
		}
	}
}
\end{algorithm}

We show that the above is a change-sensitive algorithm for enumerating $\Lambda^{del}$ in the case when the number of edges $\rho$ in $H$ is a constant.
\begin{lemma}
\label{thm:cssub}
Algorithm $\cssub$ (Algorithm~\ref{algo:cs-sub1}) enumerates all cliques in $\Lambda^{del} = \cliques(G) \setminus \cliques(G')$ using time $O(2^{\rho} |\Lambda^{new}|)$. The space complexity of the algorithm is $O(|E(G')| + |V(G')| + |\cliques(G)|)$. The algorithm can also be adapted to run in time $O(2^{\rho} |E(G)| |\Lambda^{new}|)$, and space $O(|E(G')| + |V(G')|$. 
\end{lemma}

\begin{proof}
We first show that every clique $c'$ enumerated by the algorithm is indeed a clique in $\Lambda^{del}$. To see this, note that $c'$ must be a maximal clique in $G$, due to explicitly checking the condition. Further, $c'$ is not a maximal clique in $G'$, since it is a proper subgraph of $c$, a maximal clique in $G'$. Next, we show that all cliques in $\Lambda^{del}$ are enumerated. Consider any subsumed clique $c_1' \in \Lambda^{del}$. It must be contained within $c_1-H$, where $c_1 \in \Lambda^{new}$. Moreover, $c_1'$ will be a maximal clique within $c_1-H$, and will be enumerated by the algorithm according to Lemma~\ref{lem:cliques_in_cminusH}.

For the time complexity we show that for any $c \in \Lambda^{new}$, the maximum number of maximal cliques in $c^{-H} = c-H$ is $2^{\rho}$. Proof is by induction on $\rho$. Suppose $\rho=1$ so that $H$ is a single edge, say $e_1=\{u,v\}$. Then clearly $c^{-H}$ has two maximal cliques, $c \setminus \{u\}$ and $c \setminus \{v\}$, proving the base case. Suppose that for any set $H$ of size $k$, it was true that $c^{-H}$ has no more than $2^k$ maximal cliques. Consider a set $H'' = \{e_1,e_2,\ldots,e_{k+1}\}$ with $(k+1)$ edges. Let $H' = \{e_1,e_2,\ldots,e_k\}$. Subgraph $c-H''$ is obtained from $c-H'$ by deleting a single edge $e_{k+1}$. By induction, we have that $c-H'$ has no more than $2^k$ maximal cliques. Each maximal clique $c'$ in $c-H'$ either remains a maximal clique within $c-H''$ (if at least one endpoint of $e_{k+1}$ is not in $c'$) , or leads to two maximal cliques in $c-H''$ (if both endpoints of $e_{k+1}$ are in $c'$). Hence, the number of maximal cliques in $c-H''$ is no more than $2^{k+1}$, completing the inductive step. 

\remove{Though the algorithm is written to enumerate all possible maximal cliques in $c \setminus H$, we can equivalently try all the $2^{\rho}$ potential maximal cliques in $c \setminus H$ and check for maximality in $G$. We take this approach in deriving bounds on the runtime of the algorithm.} Thus, for each cliques $c \in \Lambda^{new}$, we need to check maximality for no more than $2^{\rho}$ cliques in $G$. Note that a clique $c'$ is maximal in $G$ if it is contained in $\cliques(G)$, the set of maximal cliques in $G$. This can be done in constant time by storing the signatures of maximal cliques and checking if the signature of $c'$ is in the set of signatures of maximal cliques of $G$.


For the space bound, we first note that all operations in Algorithm~\ref{algo:cs-sub1} except maximality check \remove{including computing $c^{-H}$, enumerating maximal cliques within it, and checking for maximality,} can be done in space linear in the size of $G'$. For maximality check we need space $O(|\cliques(G)|)$ as we need to store the (signatures of) maximal cliques of $G$. The only remaining space cost is the size of $\Lambda^{new}$, which can be large. Note that the algorithm only iterates through $\Lambda^{new}$ in a single pass. If elements of $\Lambda^{new}$ were provided as a stream from the output of an algorithm such as $\csnew$, then they do not need to be stored within a container, so that the memory cost of receiving $\Lambda^{new}$ is reduced to the cost of storing a single maximal clique within $\Lambda^{new}$ at a time.

An alternative algorithm does not store $\cliques(G)$ (or hashes of elements in $\cliques(G)$). Instead, each time a potential subsumed clique $c'$ is generated that is contained in a new clique $c \in \Lambda^{new}$, we simply check $c'$ for maximality in $G$. This can be done in time $O(|E(G)|)$, by checking the intersections of the different vertex neighborhoods -- typical runtime for maximality checking can be much smaller.
\qed
\end{proof}
\subsection{Decremental Case}
\label{sec:dec}
We next consider the case when a set of edges $H$ is deleted from $G$, as opposed to added to $G$. We start from graph $G$ and go to graph $G-H$, and we are interested in efficiently enumerating $\Lambda(G,G-H)$. The decremental case can be reduced to the incremental case through the following observation. 
\begin{observation}
\label{obs:reduction}
$\Lambda^{del}(G, G-H) = \Lambda^{new}(G-H, G)$
and
$\Lambda^{new}(G, G-H) = \Lambda^{del}(G-H,G)$
\end{observation}

\begin{proof}
Consider the first equation: $\Lambda^{del}(G, G-H) = \Lambda^{new}(G-H, G)$.  Let $c \in \Lambda^{del}(G, G-H)$. This means that $c \in \cliques(G)$ and $c \not\in \cliques(G-H)$. Equivalently, $c$ is not a maximal clique in $G-H$, but upon adding $H$ to $G-H$, $c$ becomes a maximal clique in $G$. Hence, it is equivalent to say that $c \in \Lambda^{new}(G-H,G)$. Hence, we have $\Lambda^{del}(G, G-H) = \Lambda^{new}(G-H, G)$. 
The other equation, $\Lambda^{new}(G, G-H) = \Lambda^{del}(G-H,G)$, can be proved similarly.
\qed\end{proof}

\begin{algorithm}[t]
\DontPrintSemicolon
\caption{$\dec(G, H)$}
\label{algo:dec}
\KwIn{$G$ - Input Graph, $H$ - Set of $\rho$ edges being deleted}
\KwOut{All cliques in $\Lambda^{new}(G,G-H)\cup\Lambda^{del}(G,G-H)$}
$\Lambda^{new}\gets \emptyset$, $\Lambda^{del} \gets \emptyset$, $G'' \gets G - H$\;
$\Lambda^{del}\gets\csnew(G'',H)$\;
$\Lambda^{new}\gets\cssub(G'',H,\cliques(G''),\Lambda^{del})$\;
\end{algorithm}

The decremental case is outlined in Algorithm~\ref{algo:dec}.

\paragraph{Fully Dynamic Case:} Consider the {\em fully dynamic case}, where there is a set of insertions (edge set $H$) as well as deletions (edge set $H'$) from a graph. This can be processed as follows. First, we ensure there is no overlap between $H$ and $H'$, i.e. $H \cap H' = \emptyset$. If this is not the case, we can simply remove overlapping elements since they have no effect on the final graph. Next, we enumerate the change following all the edge deletions, followed by enumerating the change upon edge insertions. Note however, that this may not lead to a change-sensitive algorithm. Intermediate cliques that are output may not be in the final set of new or subsumed cliques.



%% file: expts.tex
\section{Experimental Evaluation}
\label{sec:expts}

In this section, we present results from empirical evaluation of the performance of algorithms proposed in this paper. We address the following questions: (1)~What is the runtime and memory usage for maintaining the set of maximal cliques of a dynamic graph? (2)~How does the runtime compare with the magnitude of the change? (3)~How do our algorithms compare with prior work?

\subsection{Datasets}
\textbf{Real Dynamic Graphs: }We consider graphs from the Stanford large graph database~\cite{JA14} and KONECT- The Koblenz Network Collection~\footnote{\url{http://konect.uni-koblenz.de/}}: {\tt dblp-coauthor} is a co-authorship network where each vertex represents an author and there is an edge between two authors if they have a common publication. {\tt flickr-growth} is a social network of Flickr users where each vertex represents a user and there exists a directed edge if two users are friends. {\tt sx-stackoverflow-a2q} is a social network where each vertex represents a user on stackoverflow, and if user $a$ answers user $b$'s question then there is a directed edge from $a$ to $b$. {\tt wiki-talk} is a network of Wikipedia users where each vertex represents a user and if user $a$ edited user $b$'s talk page then there exists a directed edge from $a$ to $b$. {\tt wikipedia-growth} is a hyperlink network of the English Wikipedia where each vertex represents a wikipedia page and there is an edge from a page $wiki_1$ to a page $wiki_2$ if there is a hyperlink of $wiki_2$ from $wiki_1$. {\tt youtube-u-growth} is a social network of youtube users, where nodes are the users and there is an edge between two users if they are friends. In each graph, edges have time-stamps of creation. We convert all these graphs into simple undirected graphs. If there are multiple time-stamp edges between two vertices, we take the edge with the earliest time-stamp. A summary of the graphs used in this experiment is given in Table~\ref{graph_summary}. In our experiments, we start with the empty graph and at each iteration, we add a batch of new edges, and enumerate the change resulting after the addition.

\textbf{Synthetic Graphs: } We also considers a variant of the Erd\H{o}s-R\'enyi random graph model $G(n,N)$ graph for our experiments where $n$ is the number of vertices and $N$ is the number of edges. In these, we first generate graphs according to the standard Erd\H{o}s-R\'enyi random graph model~\cite{ER}, and we ``plant'' cliques of a certain size. We call these graphs {\tt ER-1M-20M} with $1$M vertices, $20$M edges and {\tt ER-2M-15M} with $2$M vertices and $15$M edges. We plant $10$ random cliques each of size $20$ on {\tt ER-1M-20M} and $10$ random cliques each of size $30$ on {\tt ER-2M-15M}, with the goal of finding the planted cliques through incremental computation. 

\begin{table*}[t!]
\centering
\begin{tabular}{l c c}
\toprule
\textbf{Dataset} & Nodes & Edges\\
\midrule
{\tt dblp-coauthor} & $1282468$ & $5179996$\\
{\tt flickr-growth} & $2302925$ & $22838276$\\
{\tt sx-stackoverflow-a2q} & $2433067$ & $15079969$\\
{\tt wiki-talk} & $1094018$ & $2722029$\\
{\tt wikipedia-growth} & $1870709$ & $36532531$\\
{\tt youtube-u-growth} & $3223585$ & $9375374$\\
{\tt ER-1M-20M} & $1000000$ & $20001900$\\
{\tt ER-2M-15M} & $2000000$ & $15004350$\\
\bottomrule
\end{tabular}
\caption{\textbf{Summary of Graphs Used.}}
\label{graph_summary}
\end{table*}

\remove{
\begin{table*}[t!]
\centering
\begin{tabular}{l c c c}
\toprule
\textbf{Dataset} & Nodes & Edges & Maximum Degree\\
\midrule
{\tt dblp-coauthor} & $1282468$ & $5179996$ & $1522$\\
{\tt flickr-growth} & $2302925$ & $22838276$ & $27937$\\
{\tt sx-stackoverflow-a2q} & $2433067$ & $15079969$ & $22237$ \\
{\tt wiki-talk} & $1094018$ & $2722029$ & $139329$\\
{\tt wikipedia-growth} & $1870709$ & $36532531$ & $226073$\\
{\tt youtube-u-growth} & $3223585$ & $9375374$ & $91751$\\
{\tt ER-1M-20M} & $1000000$ & $20001900$ & $81$\\
{\tt ER-2M-15M} & $2000000$ & $15004350$ & $59$\\
\bottomrule
\end{tabular}
\caption{\textbf{Summary of Graphs Used.}}
\label{graph_summary}
\end{table*}
}

\subsection{Experimental Setup and Implementation Details}
We implemented all the algorithms in Java on a 64-bit Intel(R) Xeon(R) CPU with $8$G DDR$3$ RAM with $6$G heap memory.

\textbf{Algorithms:}~We evaluate our algorithm $\sdiff$ for maintenance of maximal cliques. $\sdiff$ consists of $\csnewttt$ for enumerating new maximal cliques and $\cssub$ for enumerating subsumed maximal cliques. We also implemented $\csnew$ for enumerating new maximal cliques, but $\csnewttt$ performed better, hence we present results for $\csnewttt$.

\begin{figure*}[t!]
\centering
\begin{tabular}{cc}
	\includegraphics[width=.45\textwidth]{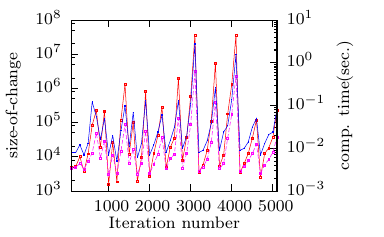} &
	\includegraphics[width=.45\textwidth]{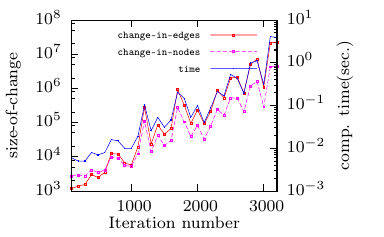} \\
	(a) {\tt dblp-coauthor} &
	(b) {\tt flickr-growth} \\
	\includegraphics[width=.45\textwidth]{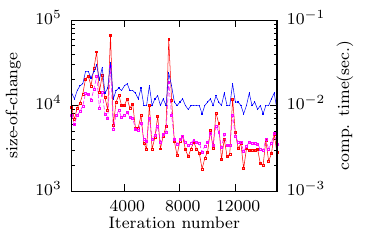} &
	\includegraphics[width=.45\textwidth]{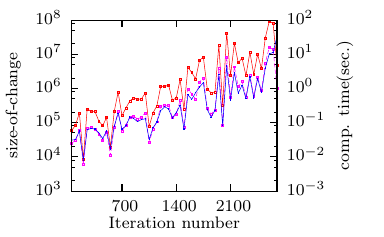} \\
	(c) {\tt sx-stackoverflow-a2q} &
	(d) {\tt wiki-talk} \\
	\includegraphics[width=.45\textwidth]{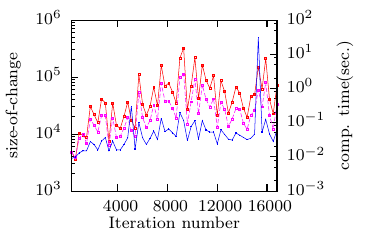} &
	\includegraphics[width=.45\textwidth]{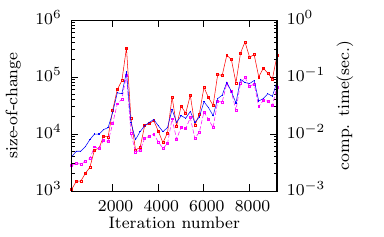} \\
	(e) {\tt wikipedia-growth} &
	(f) {\tt youtube-u-growth} \\
\end{tabular}
\caption{\textbf{Computation time for enumerating the change in set of maximal cliques for $\sdiff$, and size-of-change per batch (batch size $\rho = 1000$). The left $y$ axis shows the size of change and the right $y$ axis shows the computation time in seconds.}}
\label{fig:time_for_total_change}
\end{figure*}

\begin{table*}[t!]
\centering
\begin{tabular}{l c c c c }
\toprule
\textbf{Dataset} & $\stix$ & $\ov$ & $\mcmei$ & $\sdiff$ \\
\midrule
{\tt dblp-coauthor} ($464$) & $3811$ & $285$ & $7237$ & \textbf{0.1}\\
{\tt flickr-growth} ($251$) & $3664$ & $277$ & $7255$ & \textbf{0.04}\\
{\tt sx-stackoverflow-a2q} ($190$) & $4883$ & $232$ & $7316$ & \textbf{0.1}\\
{\tt wiki-talk} ($113$) & $7284$ & $62$ & $1425$ & \textbf{0.1}\\
{\tt wikipedia-growth} ($305$) & $3923$ & $283$ & $7190$ & \textbf{0.3}\\
{\tt youtube-u-growth} ($172$) & $3976$ & $279$ & $7257$ & \textbf{0.1}\\
\bottomrule
\end{tabular}
\caption{\textbf{Comparison of different algorithms showing cumulative time (in sec.). The number of batches for which the cumulative time is computed is in the parenthesis. Batch size $\rho$ is set to $100$.}}
\label{algorithm_comparison}
\end{table*}

\begin{figure*}[t!]
\centering
\begin{tabular}{cc}
	\includegraphics[width=.45\textwidth]{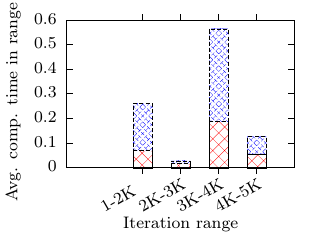} &
	\includegraphics[width=.45\textwidth]{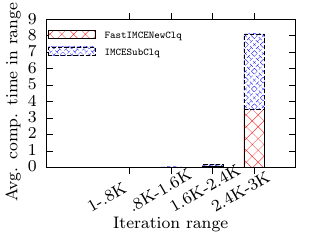} \\
	(a) {\tt dblp-coauthor} &
	(b) {\tt flickr-growth} \\
	\includegraphics[width=.45\textwidth]{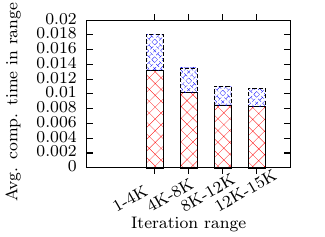} &
	\includegraphics[width=.45\textwidth]{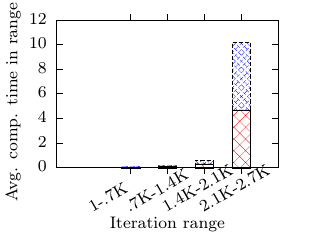} \\
	(c) {\tt sx-stackoverflow-a2q} &
	(d) {\tt wiki-talk} \\
	\includegraphics[width=.45\textwidth]{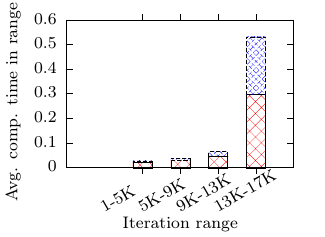} &
	\includegraphics[width=.45\textwidth]{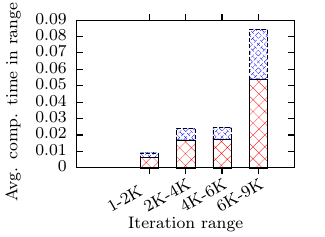} \\
	(e) {\tt wikipedia-growth} &
	(f) {\tt youtube-u-growth} \\
\end{tabular}
\caption{\textbf{Computation time (in sec.) broken down into time for new and subsumed cliques with batch size $\rho = 1000$. Average time in the $y$-axis is the average taken over the total computation times (new + subsumed) of the iterations in each of the ranges on the $x$-axis.}}
\label{fig:time_division}
\end{figure*}
\begin{table*}[t!]
\centering
\begin{tabular}{l c c}
\toprule
\textbf{Dataset} & $\csnew$ & $\csnewttt$\\
\midrule
{\tt dblp-coauthor}($9602$) & $6768$ & $19$\\
{\tt flickr-growth}($24860$) & $7318$ & $115$\\
{\tt sx-stackoverflow-a2q}($150800$) & $2852$ & $65$\\
{\tt wiki-talk}($12777$) & $7154$ & $75$\\
{\tt wikipedia-growth} ($26795$) & $562$ & $33$\\
{\tt youtube-u-growth} ($59814$) & $2981$ & $82$\\
\bottomrule
\end{tabular}
\caption{\textbf{Cumulative computation time (in sec.) for new maximal cliques with batch size $\rho = 100$. The number of batches for which the cumulative time is computed is in the parenthesis.}}
\label{csnew-vs-csnewttt}
\end{table*}

\remove{
\begin{table*}[t!]
\centering
\begin{tabular}{l c c c c c}
\toprule
\textbf{Dataset} & $\stix$ & $\ov$ & $\mcmei$ & $\naive$ & $\sdiff$ \\
\midrule
{\tt dblp-coauthor} ($464$) & $3811$ & $285$ & $7237$ & $1061$ & \textbf{0.1}\\
{\tt flickr-growth} ($251$) & $3664$ & $277$ & $7255$ & $1117$ & \textbf{0.04}\\
{\tt sx-stackoverflow-a2q} ($190$) & $4883$ & $232$ & $7316$ & $1007$ & \textbf{0.1}\\
{\tt wiki-talk} ($113$) & $7284$ & $62$ & $1425$ & $208$ & \textbf{0.1}\\
\bottomrule
\end{tabular}
\caption{\textbf{Comparison of different algorithms showing cumulative time (in sec.). The number of batches ($\rho = 100$) for which the cumulative time is computed is in the parenthesis.}}
\label{algorithm_comparison}
\end{table*}
}

\begin{table*}[t!]
\centering
\begin{tabular}{l c c}
\toprule
\textbf{Dataset} & $\sdiff$ & $\csnew$\\
\midrule
{\tt ER-1M-20M} & $19$ min. & $24$ min.\\
{\tt ER-2M-15M} & $15$ min. & $15$ min.\\
\bottomrule
\end{tabular}
\caption{\textbf{Total time taken to find all the planted cliques incrementally ($\rho = 100$). Other algorithms ($\stix$, $\ov$, $\mcmei$) cannot find a single planted clique even in an hour.}}
\label{results-synthetic}
\end{table*}

\begin{figure*}[t!]
  \def\width{0.45}
  \def\widthI{0.9}
  \centering
  \begin{subfigure}[b]{\width\textwidth}
    \centering
    \includegraphics[width=\widthI\textwidth]{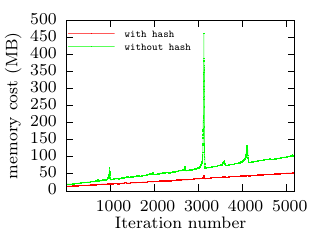}
    \caption{{\tt dblp-coauthor}}
    \label{fig:time1}
  \end{subfigure}
  \begin{subfigure}[b]{\width\textwidth}
    \centering
    \includegraphics[width=\widthI\textwidth]{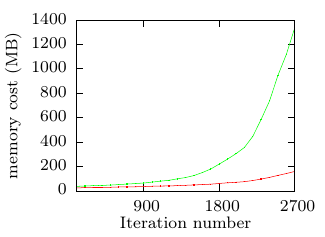}
    \caption{{\tt flickr-growth}}
    \label{fig:time2}
  \end{subfigure}
  \\
  \begin{subfigure}[b]{\width\textwidth}
    \centering
    \includegraphics[width=\widthI\textwidth]{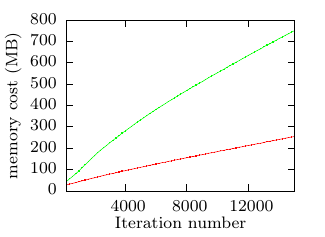}
    \caption{{\tt sx-stackoverflow-a2q}}
    \label{fig:time3}
  \end{subfigure}
  \begin{subfigure}[b]{\width\textwidth}
    \centering
    \includegraphics[width=\widthI\textwidth]{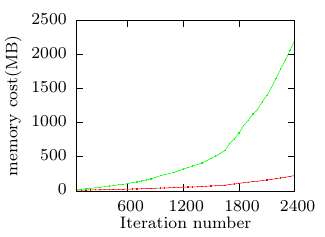}
    \caption{{\tt wiki-talk}}
    \label{fig:time4}
  \end{subfigure}
	\\
\begin{subfigure}[b]{\width\textwidth}
    \centering
    \includegraphics[width=\widthI\textwidth]{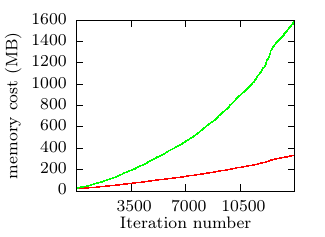}
    \caption{{\tt wikipedia-growth}}
    \label{fig:time3}
  \end{subfigure}
  \begin{subfigure}[b]{\width\textwidth}
    \centering
    \includegraphics[width=\widthI\textwidth]{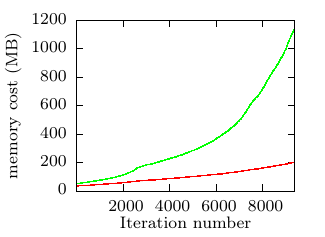}
    \caption{{\tt youtube-u-growth}}
    \label{fig:time4}
  \end{subfigure}
  \caption{Memory cost of $\sdiff$ with and without using hash function ($\rho = 1000$).}
  \label{fig:memory_consumption}
\end{figure*}


We consider the following prior algorithms for comparison with $\sdiff$: (1) $\stix$ (Stix~\cite{S04}) computes on a dynamic graph by incrementally adding one edge at a time; (2) $\ov$ (Ottosend and Vomlel~\cite{OV10}) computes on a dynamic graph by incrementally adding a set of edges; (3) $\mcmei$ (Sun et al.~\cite{SW+17}) computes on a dynamic graph by incrementally adding one edge at a time. For the algorithms ($\stix$, $\mcmei$) that support only single edge additions, we simulate the addition of a batch of edges by inserting the edges one at a time. 

\textbf{Metrics:}~We evaluate the performance of algorithms through the following metrics: \textbf{(1)}~total computation time for determining new maximal cliques and subsumed maximal cliques when a batch of new edges is added to the graph; \textbf{(2)}~change-sensitiveness, i.e, total computation time as a function of the size of the total change. We define the size of the total change in terms of edges denoted as change-in-edges as the cumulative sum of the number of edges in the new and subsumed maximal cliques. For example, if there are two new maximal cliques of sizes $3$ and $4$, and one subsumed clique of size $2$, the number of edges of a new maximal clique of size $3$ is ${3 \choose 2} = 3$, the number of edges of another new maximal clique of size $4$ is ${4 \choose 2} = 6$, and the number of edges of the subsumed clique of size $2$ is ${2 \choose 2} = 1$. Hence, the size of total change is $3 + 6 + 1 = 10$. We also consider the size of change in terms of nodes denoted as change-in-nodes where we compute the cumulative number of nodes of all cliques in the change set; \textbf{(3)}~memory cost, which includes the space required to store the graph as well as additional data structures used by the algorithm; and \textbf{(4)}~cumulative computation time (through a series of incremental updates) as a function of the size of the batch.

\subsection{Discussion of Experimental Results}
\textbf{Computation time:}~Figure~\ref{fig:time_for_total_change} shows the computation time of $\sdiff$ for computing the change in the set of maximal cliques when batches of edges are added. The batch size is set to $\rho=1000$. On the left y-axis is shown the size of the change, and on the right y-axis is the time for computing the change. We see that the runtime for computing the change in cliques becomes greater as iterations progress for graphs {\tt flickr-growth, wiki-talk, youtube-u-growth}, and remains roughly the same for other graphs. Fig.~\ref{fig:time_division} shows the breakdown of computation time of $\sdiff$ into computation time for new maximal cliques ($\csnewttt$) and computation time for subsumed maximal cliques ($\cssub$).

We also compare the computation time of $\sdiff$ with prior works as shown in Table~\ref{algorithm_comparison}. Clearly, $\sdiff$ is many orders of magnitude (more than $1000$) faster than prior algorithms. One reason why $\sdiff$ is so much faster than prior works is that $\sdiff$ systematically selects a local subgraph of the entire graph to search for new and subsumed maximal cliques. This reduces the computation effort considerably. $\ov$ tried to achieve such a local computation but $\ov$ is not provably change-sensitive for new maximal cliques, and its computation of subsumed cliques is expensive since the algorithm iterates over the entire set of maximal cliques for deriving subsumed cliques. A similar strategy of iterating over the entire set of maximal cliques for deriving maximal clique set of the updated graph as in $\mcmei$ makes the algorithm less efficient.
\begin{table*}[t!]
\centering
\begin{tabular}{l c c c c c}
\toprule
\textbf{Dataset} & $\rho = 1$ & $\rho = 10$ & $\rho = 100$ & $\rho = 1000$ & $\rho = 3\log_2 \Delta$\\
\midrule
{\tt dblp-coauthor} ($5179996$) & $1622$ & $1317$ & $1166$ & $1264$ & $1204$\\
{\tt flickr-growth} ($3298\times 10^3$) & $7151$ & $7125$ & $5551$ & $7177$ & $7000$\\
{\tt sx-stackoverflow-a2q} ($15079969$) & $270$ & $270$ & $166$ & $204$ & $106$\\
{\tt wiki-talk} ($2717\times 10^3$) & $6572$ & $6873$ & $5795$ & $6897$ & $5722$\\
{\tt wikipedia-growth} ($17000\times 10^3$) & $8869$ & $9093$ & $9134$ & $8495$ & $8678$\\
{\tt youtube-u-growth} ($9375374$) & $459$ & $462$ & $494$ & $400$ & $493$\\
\bottomrule
\end{tabular}
\caption{\textbf{Cumulative computation time (in sec.) of $\sdiff$ with different batch sizes. Note that $\Delta$ is the maximum degree of the graph before update. Numbers in the parenthesis indicates the total number of edges inserted incrementally.}}
\label{diff_batch_size}
\end{table*}

Next, we compare the computation times of $\csnew$ and $\csnewttt$ as shown in Table~\ref{csnew-vs-csnewttt}. We observe that $\csnewttt$ is much faster than $\csnew$. The increase in speed of $\csnewttt$ over $\csnew$ can be attributed to the additional pruning performed in $\csnewttt$ using $\tomitaE$, when compared with $\csnew$ which may enumerate the same clique multiple times (though suppressing it in the output).

On synthetic graphs, we observe that  $\sdiff$ can find all ``planted'' cliques in approximately $20$ min. where as the other algorithms ($\stix$, $\ov$, $\mcmei$) could not find a single planted clique in an hour. Results are shown in Table~\ref{results-synthetic}.

\textbf{Change-Sensitiveness: }
The change in the enumeration time as a function of the size of change can be seen in Figure~\ref{fig:time_for_total_change}. As the iterations progress, the size of change per batch of new edges increases for most graphs, but not in a smooth manner. In general, we observe that the time to compute the change tracks the size of change quite closely, except for graphs {\tt sx-stackoverflow-a2q, wikipedia-growth, youtube-u-growth}.

For these three graphs, we examined the breakdown of runtime more closely.  The time for handling a set of edges consists of three components: (1)~graph update time; (2)~subgraph computation time (isolating the subgraph for computing new maximal cliques); (3)~computation of new and subsumed cliques. Note that the time for the first two components do not depend on the size of change in the set of maximal cliques. In {\tt sx-stackoverflow-a2q}, the first two components take around $45\%$ of the time, and in each of {\tt wikipedia-growth, youtube-u-growth}, the first two components take around $30\%$ of the time. The significant proportion of time for the first two components means that the overall time for computing the change is not very closely correlated with the size of the change.  In the other input graphs, first two components take about $2\%$ of the overall computation time and the change-sensitive behavior can be clearly observed in these plots.

\textbf{Memory Consumption:}~Figure~\ref{fig:memory_consumption} shows the main memory used by $\sdiff$. For this experiment, we consider two different versions of the algorithm -- one with storing the clique set explicitly, and one version with only storing the hashes of the cliques. As expected, the use of a hash function reduces the memory consumption considerably. The difference in memory consumption between the two versions is especially visible in graphs {\tt flickr-growth, wiki-talk} and {\tt youtube-u-growth}, where the sizes of the maximal cliques are considerably larger. We used the $64$-bit {\tt murmur}\footnote{https://sites.google.com/site/murmurhash/} hash function on the canonical string representation of a clique, for computing the hash signature. Note that there are some ``spikes'' in the plot for {\tt dblp-coauthor}, where the memory consumption suddenly increased. On this graph, we observed that the number of maximal cliques at the point corresponding to the spike in memory usage also increased suddenly and then subsequently decreased. 

\textbf{Cumulative Computation time vs. batch size:}~ We also studied the effect of the batch size ($\rho$) on the cumulative computation time of $\sdiff$, while keeping the total number of edges added the same. For example a total of 10,000 edges would lead to 1000 batches if we used a batch size of 10, and 100 batches if we used a batch size of 100. Table~\ref{diff_batch_size} shows the results for different batch sizes. There is no observable trend found by varying the batch size.

\textbf{Summary of Results:} To summarize the results of our experiments, we note the following: (1)~$\sdiff$ is change-sensitive:  its runtime to enumerate the change in the set of maximal cliques is proportional to the magnitude of the change in the set of maximal cliques. (2)~$\sdiff$ is two to three orders of magnitude faster than prior algorithms (3)~the use of hash signatures for storing maximal cliques greatly reduces the memory consumption.

\remove{
We studied the computation time (of the entire stream) as a function of the batch size. It is interesting to see how the total time for computing all batches varies as a function of the batch size. From Figure~\ref{fig:batch_compare}, we see that the cumulative computation time for adding $2$M new edges incrementally improves as we increase the batch size. For example, choosing batch size of $10$ instead of $1$ improves in cumulative computation time by around $5$ times. We observe this difference because, increasing the size of batch reduces many intermediate computations of maximal cliques which no longer remain maximal as more batches are added in subsequent iterations.
}